%% file: root.tex
\documentclass[letterpaper, 10 pt, conference]{ieeeconf}  

\IEEEoverridecommandlockouts                              

\usepackage[utf8]{inputenc}
\usepackage{amsmath}
\usepackage{amsfonts}
\usepackage{amssymb}
\usepackage{graphicx}
\usepackage{color}
\usepackage{enumerate}
\usepackage{units}
\usepackage[style=base]{subcaption}
\usepackage{url}
\usepackage[titlenumbered,ruled,linesnumbered]{algorithm2e}

\usepackage{amsthm}

\usepackage[hidelinks]{hyperref}
\usepackage[capitalise]{cleveref}

\DeclareCaptionLabelSeparator{periodspace}{.\quad}
\theoremstyle{plain}
\newtheorem{theorem}{Theorem}
\newtheorem{lemma}{Lemma}

\newtheorem{assumption}{Assumption}

\crefname{equation}{}{} 
\crefname{section}{Sec.}{Sec.}

\newcommand{\T}{\mathrm{T}}                     

\newcommand{\ti}{\ensuremath{n}}           

\newcommand{\fctr}{\ensuremath{f_\pi}}
\newcommand{\gctr}{\ensuremath{g_\pi}}
\newcommand{\policy}{\ensuremath{\pi(x)}}
\newcommand{\argmax}{\operatornamewithlimits{argmax}}

\newif\ifextended
\extendedtrue

\newif\ifconferenceversion
\ifextended\conferenceversionfalse\else\conferenceversiontrue\fi
\newcommand{\extendedversion}[1]{\ifextended #1 \fi}
\newcommand{\conferenceversion}[1]{\ifconferenceversion #1 \fi}

\graphicspath{{figures/}}
\title{\LARGE \bf
Safe Learning of Regions of Attraction for Uncertain, Nonlinear Systems with Gaussian Processes
}

\author{Felix Berkenkamp, Riccardo Moriconi, Angela P. Schoellig, and Andreas Krause
\thanks{Felix Berkenkamp and Andreas Krause are with the Learning \& Adaptive Systems Group (\url{www.las.ethz.ch}), Department of Computer Science, ETH Zurich, Switzerland.
Email: \mbox{\{befelix, krausea\}@ethz.ch}}%
\thanks{Riccardo Moriconi is with the Department of Mechanical Engineering, ETH Zurich, Switzerland. Email: \mbox{rmoricon@ethz.ch}}%
\thanks{Angela P. Schoellig is with the Dynamic Systems Lab (\url{www.dynsyslab.org}), Institute for Aerospace Studies, University of Toronto, Canada. Email: \mbox{schoellig@utias.utoronto.ca}}%
\thanks{This research was supported by SNSF grant {200020\_159557}, NSERC grant {RGPIN-2014-04634}, and the Connaught New Researcher Award.}%
}%

\usepackage{fancyhdr}
\newcommand{\mytitle}{\textbf{Appeared in}
\textit{Proc. of the IEEE Conference on Decision and Control, 2016, pp. 4661 -- 4666, doi: \href{https://doi.org/10.1109/CDC.2016.7798979}{10.1109/CDC.2016.7798979} }.\\[0.7em]
\copyright 2016 IEEE. Personal use of this material is permitted. Permission from IEEE must be obtained for all other uses, in any current or future media, including reprinting/republishing this material for advertising or promotional purposes, creating new collective works, for resale or redistribution to servers or lists, or reuse of any copyrighted component of this work in other works.}
\begin{document}

\maketitle
\thispagestyle{fancy}
\pagestyle{empty}

\input{sections/abstract.tex}
\input{sections/introduction.tex}
\input{sections/problem_statement.tex}
\input{sections/background.tex}
\input{sections/theory.tex}
\input{sections/experiments.tex}
\input{sections/conclusion.tex}

\bibliographystyle{IEEEtran}
\bibliography{root.bib}

\extendedversion{
  \input{sections/appendix.tex}
}

\end{document}

%% file: sections/abstract.tex

\begin{abstract}
Control theory can provide useful insights into the properties of controlled, dynamic systems. One important property of nonlinear systems is the region of attraction (ROA), a safe subset of the state space in which a given controller renders an equilibrium point asymptotically stable. The ROA is typically estimated based on a model of the system. However, since models are only an approximation of the real world, the resulting estimated safe region can contain states outside the ROA of the real system. This is not acceptable in safety-critical applications. In this paper, we consider an approach that learns the ROA from experiments on a real system, without ever leaving the true ROA and, thus, without risking safety-critical failures. Based on regularity assumptions on the model errors in terms of a Gaussian process prior, we use an underlying Lyapunov function in order to determine a region in which an equilibrium point is asymptotically stable with high probability. Moreover, we provide an algorithm to actively and safely explore the state space in order to expand the ROA estimate. We demonstrate the effectiveness of this method in simulation.
\end{abstract}

%% file: sections/introduction.tex

\section{Introduction}
\label{sec:introduction}

Estimating the region of attraction (ROA) of an equilibrium point is an important problem when analyzing nonlinear systems. It specifies the region within which a given control law renders an equilibrium point asymptotically stable. Thus, the system can be operated safely within the ROA. This is of practical importance, since controllers for safety-critical systems are required to guarantee safety over a certain domain of operation, before they can be implemented on the real system. Typically, ROAs are estimated based on a model of the system. However, due to model errors, ROAs of the real system can be drastically different, raising questions about the viability of such purely model-based methods.

In this paper, we address the problem of safely estimating the ROA from experiments on a real system, while always remaining in the real ROA with high probability. As a result, the equilibrium point of the closed-loop system is attractive for any state visited throughout the learning process. This avoids failures and makes the method applicable to safety-critical domains. To expand the estimated safe region over time, we actively select states that improve our estimate of the underlying uncertain model.

\textit{Related work:}
In the literature, ROAs are computed based on system models~\cite{Zecevic2010Control}. One typical approach is to use the level sets of a Lyapunov function in order to quantify the ROA~\cite{Khalil1996Nonlinear}. Given a Lyapunov function, the problem of finding the ROA reduces to a line search for the maximum safe level set of the Lyapunov function.  For polynomial systems, Lyapunov functions can be found efficiently by solving a system of linear matrix inequalities (LMIs)~\cite{Tibken2000Estimation}. A relaxation to Lyapunov-like functions for ROA computation is given in~\cite{Ratschan2010Providing} and a review of numerical methods to compute Lyapunov functions can be found in~\cite{Giesl2015Review}. The approach in~\cite{Amato2007Region} considered quadratic systems and the problem of testing whether a given polytope lies inside the ROA. Beyond Lyapunov functions, the ROA can be found via constraint solving~\cite{Burchardt2007Estimating} or by optimizing over state trajectories~\cite{Henrion2014Convex}.
The ROA can also be approximated using sampling~\cite{Bobiti2016Sampling-1}.

An area that is of particular relevance to this paper is the computation of ROAs for systems with uncertainties. In \cite{Trofino2000Robust}, quadratic Lyapunov functions for uncertain, linear systems were used, while~\cite{Coutinho2002Guaranteed} considered polynomial Lyapunov functions. A more general approach for systems with a polynomial dependence on uncertainties within a polytope is shown in~\cite{Topcu2010Robust}. Similar ideas are used in robust control~\cite{Zhou1998Essentials}.

All the previous methods for estimating ROAs are based on a model with fixed uncertainties. The method proposed in this paper uses learning in order to reduce the model uncertainty over time. Learning the dynamics of nonlinear systems from data has been considered before. In particular, Gaussian processes (GPs~\cite{Rasmussen2006Gaussian}) have been of interest in this area, since they provide uncertainty information about the estimated dynamic model. This can be used to derive high-probability safety guarantees. For example,~\cite{Berkenkamp2015Safe} combines GP regression with robust control theory to provide stability and performance guarantees for a linear controller applied to a nonlinear system, while~\cite{Akametalu2014Reachability} uses reachability analysis to compute robust invariant sets.

In our method, we use ideas from Bayesian optimization~\cite{Mockus1989Bayesian},  set of optimization methods that aim to find the global optimum of an~\textit{a priori} unknown function based on noisy evaluations.
Bayesian optimization algorithms that are based on GP models provably converge close to the global optimum~\cite{Srinivas2010Gaussian}. In this setting,~\cite{Sui2015Safe,Berkenkamp2016Safe} consider safety in terms of high-probability constraint satisfaction,
with a generalization to multiple constraints in~\cite{Berkenkamp2016Bayesian}.

\textit{Our contribution:}
In this paper, we combine ideas from GP learning, safe Bayesian optimization, and ROA computation based on Lyapunov functions for uncertain systems. Using previous results on model-based ROA computation as a starting point, we compute ROA estimates for the real system by approximating the model uncertainties with a GP on a discrete domain. We use ideas from safe Bayesian optimization in order to actively learn about the real dynamics from experiments executed in the estimated ROA. As we learn more about the dynamics of the real system, the uncertainty in the estimate decreases and the ROA expands, until we reach a desired accuracy. As a result, we are able to learn an accurate estimate of the ROA of a general, nonlinear system through experiments, without leaving the ROA and, thus, without violating safety requirements.

%% file: sections/problem_statement.tex

\section{Problem Statement}
\label{sec:prob_def}

We consider a nonlinear, continuous-time system,
\begin{equation}
\label{eq:dynamics_ode}
\dot{x}(t) = \underbrace{f(x(t), u(t))}_{\textnormal{\textit{a priori} model}}  + \underbrace{g(x(t), u(t)) ,}_{\textnormal{unknown model}}
\end{equation}
where~${x(t) \in \mathcal{X} \subseteq \mathbb{R}^q }$ is the state at time $t$ within a connected set~$\mathcal{X}$ and ${u(t) \in \mathcal{U} \subseteq \mathbb{R}^p}$ is the control input. The system dynamics consist of a \textit{known} model~$f(x,u)$ and an \textit{unknown} model~$g(x,u)$. The latter accounts for unknown dynamics and model uncertainties. We assume a control policy~${u = \policy}$ is given, which has been designed for the prior model,~$f(\cdot)$. The resulting, closed-loop dynamics are denoted by~${\fctr(x) := f(x, \policy)}$ and~${\gctr(x) := g(x, \policy)}$.

The goal is to estimate the ROA of~\cref{eq:dynamics_ode} under the control policy,~$\policy$, based on experiments on the real system. We want to actively learn the \textit{a priori} unknown dynamics,~$\gctr$, without leaving the ROA. This is, of course, impossible without further assumptions about the model in~\cref{eq:dynamics_ode}.

We assume that the unknown model,~$\gctr$, has low `complexity', as measured under the norm of a \textit{reproducing kernel Hilbert space} (RKHS,~\cite{Christmann2008Support}). An RKHS~${\mathcal{H}_k(\mathcal{X})}$ is a complete subspace of~$L_2(\mathcal{X})$, the space of square integrable functions, that includes functions of the form~${ \gctr(x) = \sum_{i} \alpha_i k(x, x_i) }$ with~${\alpha_i \in \mathbb{R}}$ and representer points~${x_i \in \mathcal{X}}$, where~$k$ (and the subscript~$k$) refer to a symmetric, positive-definite kernel function~${k(\cdot,\cdot)}$. The RKHS has an inner product~${\langle\cdot,\cdot\rangle_k}$, which obeys the reproducing property:~${\langle \gctr(\cdot) , k(x, \cdot) \rangle_k = \gctr(x)}$ for all~${\gctr \in \mathcal{H}_k(\mathcal{X})}$. The induced RKHS norm~${\| \gctr \|_k^2 = \langle \gctr, \gctr \rangle_k }$ measures smoothness of~$\gctr$ with respect to the kernel~${k(\cdot,\cdot)}$,~\cite{Srinivas2010Gaussian}. For universal kernels~\cite{Christmann2008Support}, members of~$\mathcal{H}_k(\mathcal{X})$ can uniformly approximate any continuous function on any compact subset of~$\mathcal{X}$.
We make assumptions about the complexity of~$\gctr$.

\begin{assumption}
\label{as:g_bounded_RKHS}
The function~$\gctr$ has bounded RKHS norm with respect to a continuously differentiable, bounded kernel~$k(x, x')$; that is,~${ \|\gctr\|_k \leq B_g }$.
\end{assumption}

Moreover, we make the following, standard assumption on the \textit{a priori} model,~$\fctr(x)$:
\begin{assumption}
\label{as:f_lipschitz}
The function $\fctr$ is Lipschitz continuous with Lipschitz constant~$L_f$ and bounded in~$\mathcal{X}$ with~${\|\fctr \|_\infty \leq B_f}$.
\end{assumption}

In the following sections, we use a Lyapunov function in order to compute the ROA. In practice, one may compute good Lyapunov functions during the experimentation phase using, for example, the methods in~\cite{Giesl2015Review}. Here, we consider a fixed, given Lyapunov function. We assume the following:
\begin{assumption}
\label{as:f_g_zero_at_origin}
The origin is an equilibrium point of~\cref{eq:dynamics_ode} with ${\fctr(0) = \gctr(0) = 0}$.
\end{assumption}
\begin{assumption}
\label{as:lyapunov_function_stable_f}
\label{thm:deterministic_c0_stable}
A fixed, two-times continuously differentiable Lyapunov function~$V(x)$ is given. Moreover, there exists a constant~${c_0 > 0}$ such that~${ \frac{\partial V(x)}{\partial x} (\fctr(x) + \gctr(x)) < 0}$ for all~${x \in \mathcal{S}_0=\mathcal{V}(c_0)}$, where~$\mathcal{V}(c) = \{x \in \mathcal{X} \,|\, V(x) \leq c\}$.
\end{assumption}
\cref{as:f_g_zero_at_origin} implies that the origin is an equilibrium point of~\cref{eq:dynamics_ode} with the unknown model~$\gctr(x)$.
Without this assumption, there is no hope to prove that the origin is attractive, but showing boundedness would be possible. \cref{as:lyapunov_function_stable_f} implies that we have chosen the policy,~$\policy$, and the Lyapunov function,~$V(x)$, such that the origin of~\cref{eq:dynamics_ode} is locally asymptotically stable. This is, for example, fulfilled if the origin of the \textit{a priori} dynamics~$\dot{x} = \fctr(x)$ is locally asymptotically stable and~$\frac{\partial \gctr(x)}{\partial x}$ is zero at the origin; that is, the prior dynamics dominate close to the origin.

Lastly, we require measurements of~$\gctr(x)$ in order to learn about the unknown dynamics. That is, we must be able to measure~$x$ and~$\dot{x}$ (${\fctr(\cdot)}$ is known). In practice, we may consider a  discrete-time approximation of $\dot{x}$, which only requires measurements of the state,~$x$, rather than its derivative, $\dot{x}$. We discuss practical aspects in~\cref{sec:discussion}.
\begin{assumption}
\label{as:measurements_of_deriv}
We have access to measurements, ${\hat{g}_\pi(x) = \dot{x} - \fctr(x) + \omega}$, which are corrupted by zero-mean, independent and bounded noise,~${\| \omega \| \leq \sigma}$.
\end{assumption}

For ease of notation, we consider a one-dimensional system of the form of~\cref{eq:dynamics_ode} in the following. We explain how to generalize the results to multiple dimensions in~\cref{sec:mtuliple_dimensions}.

%% file: sections/background.tex

\section{Gaussian Processes (GPs)}
\label{sec:background}
\label{sec:gaussian_process}

The function~$\gctr(x)$ in~\cref{eq:dynamics_ode} is unknown \textit{a priori}. Given~\cref{as:g_bounded_RKHS}, we use a GP to approximate this unknown function over its domain~$\mathcal{X}$, see~\cite{Srinivas2010Gaussian} and~\cref{thm:confidence_interval} below. The following introduction to GPs is based on~\cite{Rasmussen2006Gaussian}.

GPs are a nonparametric regression method from machine learning, where the goal is to find an approximation of a nonlinear map~${\gctr \colon \mathcal{X} \to \mathbb{R}}$ from a state~$x$ to the function value~$\gctr(x)$.
This is accomplished by considering the function values~$\gctr(x)$ to be random variables, so that any finite number of them have a joint Gaussian distribution~\cite{Rasmussen2006Gaussian}.

The Bayesian, nonparametric regression is based on a prior mean function and a covariance function,~$k(x, x')$, which defines the covariance of any two function values,~$\gctr(x)$ and $\gctr(x')$, ${x, x' \in \mathcal{X}}$.
The latter is also known as the kernel. In this work, the mean is zero, since any prior knowledge about the dynamics is captured by~$\fctr(x)$; that is,~${\gctr(x) \sim \mathcal{GP}(0,\,k(x,x'))}$.
In general, the choice of kernel function is problem-dependent and encodes assumptions about the unknown function. A discussion of kernel choices for dynamic systems can be found in~\cite{Berkenkamp2016Kernel}. In the following, we use the same kernel as in~\cref{as:g_bounded_RKHS}.

We can obtain the posterior distribution of a function value~${\gctr(x)}$ at an arbitrary state~${ x \in \mathcal{X} }$ by conditioning the GP distribution of~$\gctr$ on a set of~$\ti$ past measurements, ${y_\ti = (\hat{g}_\pi(x_1), \dots, \hat{g}_\pi(x_\ti) ) }$ at states~${ \mathcal{D}_\ti = \{x_1, \dots, x_\ti \} }$, where ${\hat{g}_\pi(x) = \dot{x} - \fctr(x) + \omega}$ with~${\omega \sim \mathcal{N}(0,\,\sigma^2)}$. The posterior over~$\gctr(x)$ is a GP distribution again, with mean~$\mu_\ti(x)$, covariance~$k_\ti(x,x')$, and variance~$\sigma_\ti(x)$:
\begin{align}
\mu_\ti(x) &= k_\ti(x) (K_\ti + \mathbb{I}_\ti \sigma^2)^{-1} y_\ti ,
\label{math:gp_prediction_mean} \\
k_\ti(x, x') &= k(x, x') - k_\ti(x) (K_\ti + \mathbb{I}_\ti \sigma^2)^{-1} k_\ti^\T(x'),
\label{math:gp_prediction_covariance} \\
\sigma^2_\ti(x) &= k_\ti(x, x),
\label{math:gp_prediction_variance}
\end{align}
where the vector
${k_\ti(x) =
\left( \begin{matrix}
	k(x, x_1), \dots, k(x, x_\ti)
\end{matrix}  \right)}$
contains the covariances between the new input,~$x$, and the states in~$\mathcal{D}_\ti$,~$K_\ti \in \mathbb{R}^{\ti \times \ti}$ is the positive-definite kernel matrix with ${[K_\ti]_{(i,j)} = k(x_i,\, x_j)}$,~${i, j \in \{1,\dots,\ti\} }$, and~${\mathbb{I}_\ti \in \mathbb{R}^{\ti \times \ti}}$ is the identity matrix.

\cref{as:g_bounded_RKHS} allows us to model~$\gctr(x)$ as a GP. In particular, we have the following result:

\begin{lemma}
\label{thm:confidence_interval}
Supposed that~${\|\gctr\|^2_k \leq B_g}$ and that the zero-mean noise~$\omega$ is uniformly bounded by~$\sigma$. Choose~${\beta_\ti = 2B_g + 300 \gamma_\ti \log^3(\ti / \delta)}$, where~$\gamma_\ti$ is defined in the paragraph below. Then, for all~${\ti \geq 1}$ and~${x \in \mathcal{X}}$, it holds with probability at least~${1 - \delta}$, ${\delta \in (0, 1)}$, that
\begin{equation}
| \gctr(x) - \mu_{\ti-1}(x) | \leq \beta_\ti^{1/2} \sigma_{\ti-1}(x).
\end{equation}
\end{lemma}
\begin{proof}
See~\cite[Theorem 6]{Srinivas2010Gaussian}.
\end{proof}

\cref{thm:confidence_interval} allows us to make high-probability statements about the function values of~$\gctr(x)$, even though the GP model makes different assumptions than we do in~\cref{as:g_bounded_RKHS} (e.g., about the noise,~$\omega$). The scalar~$\beta_\ti$ depends on the information capacity,~${\gamma_\ti = \max_{x_1, \dots, x_\ti}\,I(\gctr, y_\ti)}$, which is the maximal mutual information that can be obtained about the GP prior from~$\ti$ noisy samples~$y_\ti$ at states~$\mathcal{D}_\ti$. In~\cite{Srinivas2010Gaussian}, it was shown that~$\gamma_n$ has a sublinear dependence on~$\ti$ for many commonly used kernels and that it can be efficiently approximated up to a constant. As a result, even though~$\beta_\ti$ is increasing with~$\ti$, we are able to learn about the true values of~$\gctr(x)$ over time by making appropriate choices for data samples in $\mathcal{D}_\ti$, see~\cref{sec:active_learning}.

%% file: sections/theory.tex

\section{Lyapunov Stability}
\label{sec:lyapunov}

In this section, we show how the assumptions in~\cref{sec:prob_def} can be used in order to compute the ROA for the nonlinear system in~\cref{eq:dynamics_ode} based on a GP model of~$\gctr$ from~\cref{sec:gaussian_process}.
For now, we base our analysis on a general GP model of~$\gctr(x)$ from~\cref{sec:gaussian_process} with~$\ti$ arbitrary measurements of~$\gctr$. We actively select states for new measurements in~\cref{sec:active_learning}.

We start by observing the following:
\begin{lemma}
\label{thm:g_lipschitz}
The function $\gctr$ is Lipschitz continuous with Lipschitz constant~$L_g$ and bounded by~$\|\gctr \|_\infty \leq B_g \|k\|_\infty$.
\end{lemma}
\begin{proof}
Boundedness by~\cite[Lemma 4.23]{Christmann2008Support}.
From~\cref{as:g_bounded_RKHS}, ${|\gctr(x) - \gctr(x')|^2 \leq \|\gctr\|^2_k\, d_k(x,x')}$, where ${d_k(x,x') = k(x, x) + k(x',x') - 2k(x,x')}$ is the kernel metric~\cite[Lemma 4.28]{Christmann2008Support}. Since~$k(\cdot,\cdot)$ is continuously differentiable and bounded, ${|\gctr(x) - \gctr(x')|^2 \leq L_g^2 |x - x'|^2}$, where ${L_g^2 = 2 B_g \|k\|_\infty \| \frac{\partial k}{\partial x} \|_\infty}$.
\end{proof}

Since the closed-loop dynamics in~\cref{eq:dynamics_ode} are Lipschitz continuous based on~\cref{thm:g_lipschitz} and~\cref{as:f_lipschitz}, existence and uniqueness of the solutions of~\cref{eq:dynamics_ode} are guaranteed.

The goal of this section is to quantify the ROA based on the Lyapunov function in~\cref{as:lyapunov_function_stable_f}. From Lyapunov stability theory we have the following~\cite{Khalil1996Nonlinear}:
\begin{lemma}
\label{thm:lyapunov_stability}
The origin of the dynamics in~\cref{eq:dynamics_ode} is asymptotically stable within a level set,~${ \mathcal{V}(c) = \{ x \in \mathcal{X} \,|\, V(x) \leq c \} }$ with~${c \in \mathbb{R}_{>0}}$, if, for all~${x \in \mathcal{V}(c)}$,
\begin{equation}
\dot{V}(x)=\frac{\partial V(x)}{\partial x} \big( \fctr(x) + \gctr(x) \big) < 0.
\label{eq:lyap_stability_requirement}
\end{equation}
\end{lemma}
The goal of this section is to use the GP model of~$\gctr(x)$ from~\cref{sec:gaussian_process} in order to find the largest~$c$ such that~\cref{eq:lyap_stability_requirement} holds true within~$\mathcal{V}(c)$ with high probability. The existence of such a value~${c_0 > 0}$ is ensured by~\cref{thm:deterministic_c0_stable}.

In order to quantify the ROA using~$\cref{eq:lyap_stability_requirement}$, we need to evaluate~${\dot{V}(x)}$, which depends on the GP model of~$\gctr(x)$. From~\cref{sec:gaussian_process} we know that the GP posterior distribution over~$\gctr$ is Gaussian. As a consequence, since~${\dot{V}(x)}$ in~\cref{eq:lyap_stability_requirement} is affine in~$\gctr(x)$,~${\dot{V}(x)}$ itself is a GP. In particular, we have that ${\dot{V}(x)}$ is normally distributed with mean~${\mu_{\ti,\dot{V}}(x)}$ and standard deviation~${\sigma_{\ti,\dot{V}}(x)}$ given by:
\begin{align}
\mu_{\ti, \dot{V}}(x) &= \frac{\partial V(x)}{\partial x} \big( \mu_\ti(x) + \fctr(x) \big),
\label{eq:v_dot_mean_prediction}\\
\sigma_{\ti, \dot{V}}(x) &= \left| \frac{\partial V(x)}{\partial x} \right| \sigma_\ti(x),
\label{eq:v_dot_var_prediction}
\end{align}
where~$\mu_\ti(x)$ and~$\sigma_\ti(x)$ are the GP predictions of the unknown dynamics,~$\gctr(x)$, from~\cref{math:gp_prediction_mean,math:gp_prediction_variance}. We can use~\cref{eq:v_dot_mean_prediction,eq:v_dot_var_prediction} to directly make predictions about~$\dot{V}(x)$ in~\cref{eq:lyap_stability_requirement} based on measurements of~$\gctr(x)$.

Since~\cref{eq:v_dot_mean_prediction,eq:v_dot_var_prediction} provide a probabilistic estimate of~${\dot{V}(x)}$, we cannot expect to make deterministic statements about stability. Instead, we consider the confidence intervals of the GP model of~${\dot{V}(x)}$. We denote the lower and upper confidence intervals after~${(\ti-1)}$ measurements of~${\gctr(x)}$~by
\begin{align}
l_\ti(x) &:= \mu_{\dot{V}, \ti-1}(x) - \beta_\ti^{1/2} \sigma_{\dot{V}, \ti-1},
\label{eq:v_dot_lower_bound} \\
u_\ti(x) &:= \mu_{\dot{V}, \ti-1}(x) + \beta_\ti^{1/2} \sigma_{\dot{V}, \ti-1},
\label{eq:v_dot_upper_bound}
\end{align}
respectively. The confidence intervals are parameterized by the scalar~$\beta_\ti$. In the following, we assume that~$\beta_\ti$ is chosen according to~\cref{thm:confidence_interval}, which allows us to say that~$\dot{V}(x)$ takes values within the interval~${[l_\ti(x),\,u_\ti(x)]}$ with high probability (at least~${1-\delta}$).

Based on these confidence intervals, we can see from~\cref{eq:lyap_stability_requirement} that~${\dot{V}(x) < 0}$ holds within~$\mathcal{V}(c)$ for some~${c>0}$ with high probability, if~${ u_\ti(x) < 0 }$ for all~${x \in \mathcal{V}(c)}$. Unfortunately, it is impossible to evaluate the upper bound~\cref{eq:v_dot_upper_bound} everywhere in a continuous domain. Nevertheless, it is possible to evaluate predictions of~${ \dot{V}(x)}$ from~\cref{eq:v_dot_mean_prediction,eq:v_dot_var_prediction} at a finite number of points. In the following, we exploit the continuity properties of~${ \dot{V}(x) }$ in order to derive stability properties within a continuous domain based on a finite number of predictions. In particular, we use the following property of~${\dot{V}}$:
\begin{lemma}
\label{thm:v_dot_lipschitz}
The function ${\dot{V}(x)}$ is Lipschitz continuous with Lipschitz constant~$L$.
\end{lemma}
\begin{proof}
Based on the Lipschitz continuity of~\cref{eq:dynamics_ode} from~\cref{as:f_lipschitz,thm:g_lipschitz}, we expand~${|\dot{V}(x) - \dot{V}(x')|}$ to
\begin{align*}
&\left| \frac{\partial V(x)}{\partial x} (\fctr(x) + \gctr(x)) - \frac{\partial V(x')}{\partial x} (\fctr(x') + \gctr(x')) \right| \\
&\leq \bigg| \fctr(x) + \gctr(x) \bigg| \,  \left| \frac{\partial V(x)}{\partial x} - \frac{\partial V(x')}{\partial x'} \right| \\
&\phantom{\leq}+ \left| \frac{\partial V(x')}{\partial x'} \right|  \bigg| \fctr(x) + \gctr(x) - \fctr(x') - \gctr(x')  \bigg|, \\
&\leq (B_f + B_g \|k\|_\infty) L_{\partial V} | x - x' | + L_V (L_g+L_f) | x - x' |, \\
&:= L | x - x' |,
\end{align*}
where~${L_V = \| \frac{\partial V(x)}{\partial x} \|_\infty}$ and~${L_{\partial V} = \| \frac{\partial^2 V(x)}{\partial x^2} \|_\infty}$ are the Lipschitz constants of~$V$ and its first derivative. These are guaranteed to exist by~\cref{as:lyapunov_function_stable_f}, since a continuous function obtains a maximum over a bounded domain.
\end{proof}

The continuity of~${\dot{V}(x) }$ allows us to evaluate predictions of~$\dot{V}(x)$ from~\cref{eq:v_dot_mean_prediction,eq:v_dot_var_prediction} only at a finite number of points, without loosing guarantees.
\begin{lemma}
\label{thm:discretization}
Let~${\mathcal{X}_\tau \subset \mathcal{X}}$ be a discretization of~$\mathcal{X}$ with ${ | x - [x]_\tau | \leq \tau / 2}$
for all~${x \in \mathcal{X}}$. Here,~$[x]_\tau$ denotes the closest point in~${\mathcal{X}_\tau}$ to~${x \in \mathcal{X}}$. Choosing~${\beta_\ti}$ according to~\cref{thm:confidence_interval}, the following holds with probability at least~${(1 - \delta)}$ for all~${x \in \mathcal{X}}$ and all~$\ti \geq 1$:
\begin{equation}
\left| \dot{V}(x) - \mu_{\dot{V},\,\ti-1}([x]_\tau) \right| \leq \beta_\ti^{1/2} \sigma_{\dot{V},\,\ti-1}([x]_\tau) + L \tau.
\end{equation}
\end{lemma}
\begin{proof}
The result follows from the Lipschitz continuity of~$\dot{V}(x)$ in~\cref{thm:v_dot_lipschitz} and is similar to~\cite[Lemma 5.7]{Srinivas2010Gaussian}.
\end{proof}

\cref{thm:discretization} provides high-probability bounds on~$\dot{V}$ on the continuous domain~$\mathcal{X}$ using the GP confidence intervals~\cref{eq:v_dot_mean_prediction,eq:v_dot_var_prediction} on the discrete set~$\mathcal{X}_\tau$. It achieves this by using the Lipschitz property of~$\dot{V}$ to generalize from~$\mathcal{X}_\tau$ to~$\mathcal{X}$.
The discretization accuracy~$\tau$ trades off the accuracy for the reduced computational cost of computing the confidence intervals for all states in~$\mathcal{X}_\tau$. Combining the previous results lets us argue about the stability:
\begin{lemma}
\label{thm:stability_requirement_in_X}
With a discretization of~$\mathcal{X}$,~$\mathcal{X}_\tau$, according to~\cref{thm:discretization} and with~$\beta_\ti$ according to~\cref{thm:confidence_interval}, the origin of~\cref{eq:dynamics_ode} is asymptotically stable within~$\mathcal{V}(c)$ for some~${c > 0}$ with probability at least~${(1 - \delta)}$ if, for all~$x \in \mathcal{V}(c) \cap \mathcal{X}_\tau$,
\begin{equation}
\mu_{\dot{V},\,\ti-1}(x) + \beta_\ti^{1/2} \sigma_{\dot{V},\,\ti-1}(x) = u_\ti(x) < -L \tau.
\label{math:thm:stability_requirement}
\end{equation}
\end{lemma}
\begin{proof}
This is a consequence of~\cref{thm:confidence_interval,thm:discretization,thm:lyapunov_stability}.
\end{proof}
Given the previous results, after~${(\ti - 1)}$ evaluations of~$\gctr$, it suffices to maximize~$c$ such that the condition in~\cref{thm:stability_requirement_in_X} holds for all discretized states in~${\mathcal{V}(c)}$:
\begin{theorem}
\label{thm:roa_in_S}
Under the assumptions of~\cref{thm:stability_requirement_in_X}, pick
\begin{equation}
c_\ti = \argmax_{c \in \mathbb{R}_{>0}}\, c, \quad \textnormal{subject to~\cref{math:thm:stability_requirement} for all } x \in \mathcal{V}(c) \cap \mathcal{X}_\tau.
\label{eq:linesearch}
\end{equation}
Then,~${\mathcal{S}_\ti = \mathcal{S}_0} \cup \mathcal{V}(c_\ti)$ is contained within the ROA of~\cref{eq:dynamics_ode} for all~${\ti \geq 1}$ with probability at least~${(1 - \delta)}$.
\end{theorem}
\begin{proof}
The set~$\mathcal{S}_0$ from~\cref{thm:deterministic_c0_stable} is contained in the ROA deterministically, while~${\mathcal{V}(c_\ti)}$ is contained with probability at least~${(1-\delta)}$ by~\cref{thm:stability_requirement_in_X}. The result follows.
\end{proof}

\cref{thm:roa_in_S} provides a way to compute an estimate of the ROA via an efficient binary search to solve the optimization problem~\cref{eq:linesearch}. The discretization in~\cref{thm:stability_requirement_in_X} allows us to consider high-probability confidence intervals on~$\mathcal{X}_\tau$ only, while we generalize using the Lipschitz constant from~\cref{thm:v_dot_lipschitz}. A finer discretization (smaller value of~$\tau$) decreases the conservativeness of the ROA estimate. This is different from the number of experiments,~$\ti$, which decreases the uncertainty in the model estimate as we obtain more data.
In particular, as more information about~$\gctr$ becomes available from measurements, the ROA increases beyond the initial, deterministic safe set in~\cref{thm:deterministic_c0_stable}.

\section{ACTIVE LEARNING}
\label{sec:active_learning}

\begin{figure*}[t]
\vspace{1.5mm}
\centering
\subcaptionbox{Initial safe set,~$\mathcal{S}_0$. \label{fig:1d_example_0}}
    {\includegraphics{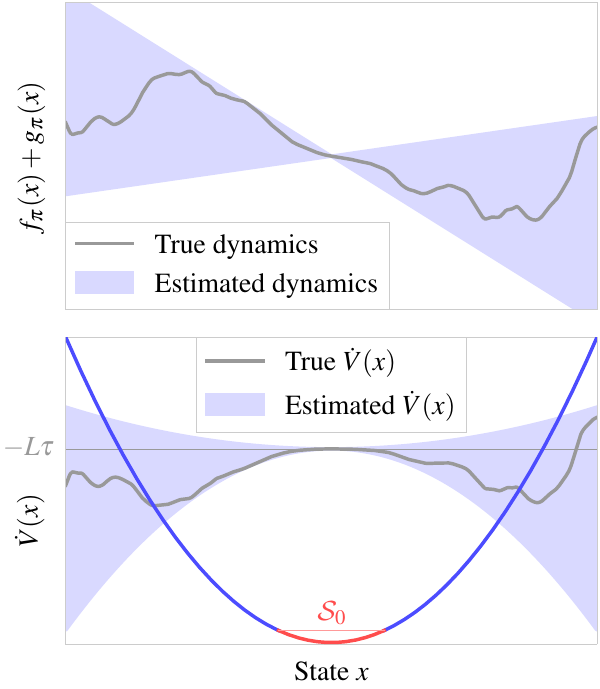}}
\subcaptionbox{Exploration within safe level set,~$\mathcal{S}_\ti$. \label{fig:1d_example_1}}
    {\includegraphics{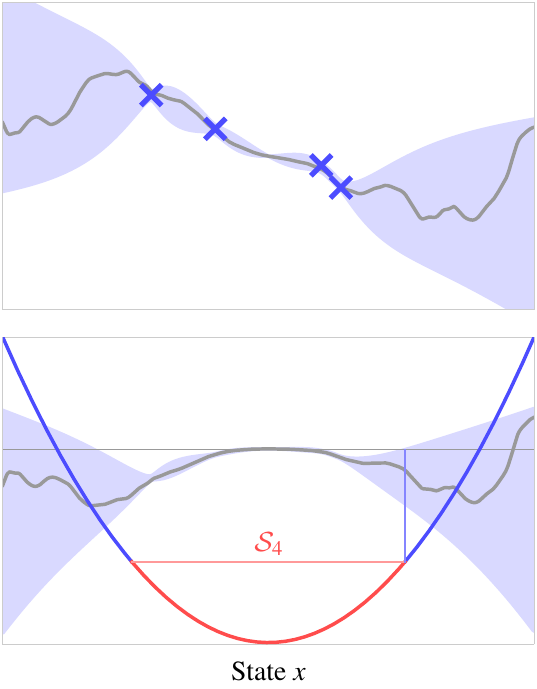}}
\subcaptionbox{Maximum level set found. \label{fig:1d_example_2}}
    {\includegraphics{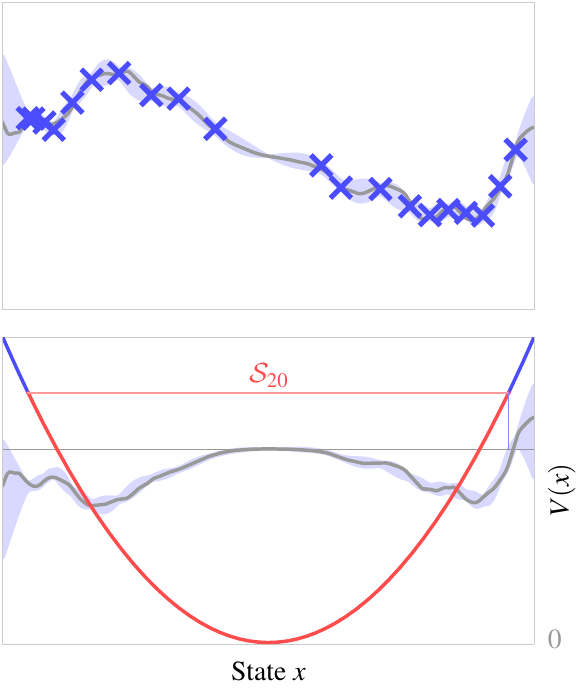}}
\caption{One-dimensional example of~\cref{alg:safe_exploration}. Initially, in~\cref{fig:1d_example_0}, the estimate of the dynamics is uncertain (top row, blue shade is the $2\sigma$ confidence interval). As a result, the ROA (red part of the Lyapunov function, bottom plot) consists only of the initial safe set,~$\mathcal{S}_0$; larger level sets of the Lyapunov function are unsafe (blue part). The algorithm actively selects new states at which to obtain measurements of the dynamics, which causes the safe set to increase in~\cref{fig:1d_example_1}. Eventually, this procedure leads to the largest safe level set in~\cref{fig:1d_example_2}. Only states inside the ROA are evaluated during the learning.}
\label{fig:1d_example}
\end{figure*}

\begin{algorithm}[t]
  \SetAlgoNoEnd
    \caption{Safe ROA exploration}
    \label{alg:safe_exploration}
    \DontPrintSemicolon
    \SetKwInOut{Input}{Inputs}
    \Input{Domain~$\mathcal{X}$ and discretization with~$\tau$,~$\mathcal{X}_\tau$ \newline
           GP prior $k(x, x')$ \newline
           Initial safe set~${\mathcal{S}_0 \subseteq \mathcal{X}}$}

    \For{$\ti = 1 , \dots $}{
    	$\begin{aligned} c_\ti \gets &\argmax_{c > 0}\,c, \,\,\, \textnormal{subject to} \\
    	&\quad u_\ti(x) < -L \tau,\, \textnormal{ for all } x \in \mathcal{V}(c)  \cap \mathcal{X}_\tau
    	\end{aligned}$ \;
		$\mathcal{S}_\ti \gets \mathcal{S}_0 \,\cup\, \mathcal{V}(c_\ti)$ \;
		$x_\ti \gets \argmax_{x \in \mathcal{S}_\ti} \, \sigma_{\ti-1}(x)$ \;
		Update GP with measurement of~$\hat{g}_\pi(x_\ti)$
    }
\end{algorithm}

In the previous section, we provided an estimate of the ROA based on a GP model of the unknown dynamics of~$\gctr(x)$ in~\cref{eq:dynamics_ode}, which was informed by arbitrarily selected measurements/experiments. In this section, we actively expand the ROA estimate by selecting new, safe states within the current ROA estimate at which to obtain measurements.

As we update the model with new data, we can compute the current estimate of the ROA,~$\mathcal{S}_\ti$, using~\cref{thm:roa_in_S}. In order to expand the ROA from the initial estimate~$\mathcal{S}_0$, we need to obtain measurements at states that are relevant for learning~$\gctr$. Here, we assume that we have access to a control method that drives us to desired states without leaving the ROA. Any method that applies the policy~$\pi(x)$ on the boundary of~$\mathcal{S}_n$ fulfills this safety requirement~\cite{Akametalu2014Reachability}. Based on ideas from safe Bayesian optimization~\cite{Sui2015Safe}, we use the uncertainty estimate of the GP. We select
\begin{equation}
x_\ti = \argmax_{x \in \mathcal{S}_\ti} \, \sigma_{\ti-1}(x)
\label{eq:acquisition_function}
\end{equation}
as the next state to evaluate according to~\cref{math:gp_prediction_variance} within a current estimate of the safe set,~$\mathcal{S}_\ti$. At~$x_\ti$, the GP model is the most uncertain about the unknown dynamics,~$\gctr(x)$. The idea behind this selection criterion is that we want to learn about the most uncertain state in order to increase the ROA estimate. By decreasing the uncertainty about the unknown dynamics,~$\gctr(x)$, over time, the safe region expands eventually.
We summarize the method in~\cref{alg:safe_exploration} and an example is shown in~\cref{fig:1d_example}.

It follows from~\cref{thm:roa_in_S} that every state chosen by~\cref{eq:acquisition_function} lies inside the ROA of~\cref{eq:dynamics_ode} with high probability. For a variant of~\cref{alg:safe_exploration} it is possible to prove an even stronger result: that the maximum, safe level set can be found to some accuracy (with probability at least~${(1-\delta)}$) after a finite number of experiments.
\extendedversion{The corresponding variant of~\cref{alg:safe_exploration} builds up the ROA estimate using the Lipschitz constant. More details are given in the appendix in~\cref{thm:full_exploration}.}
\conferenceversion{More details are given in the extended version of the paper at~\href{http://arxiv.org/abs/1603.04915}{http://arxiv.org/abs/1603.04915}.}

\section{Extension to Multiple Dimensions}
\label{sec:mtuliple_dimensions}

So far, we have considered a one-dimensional system in~\cref{eq:dynamics_ode}. The preceding analysis directly transfers to multiple dimensions, at the expense of more cumbersome notation. The only non-trivial parts of the extension are GP models with vector predictions and the choice of~$\beta_\ti$ in~\cref{thm:confidence_interval}. The main observation is that vector-valued functions can be modeled as a single GP over an extended state space,~${\mathcal{X} \times \mathcal{I}}$, where integer elements in~$\mathcal{I}$ index the output dimension of~$\gctr(x)$. At each iteration, we obtain~$q$ measurements of~$\gctr(x)$; one for each dimension of the state. As a result,~${\beta_{\ti'} = 2B_g + 300\gamma'_{\ti q} \log^3(\ti q / \delta)}$ increases at a faster rate of~${\ti'=\ti q}$ compared to~${\beta_\ti}$ in~\cref{thm:confidence_interval}. The information capacity,~$\gamma'$, is measured relative to the combined function over~${\mathcal{X} \times \mathcal{I}}$. Refer to~\cite{Berkenkamp2016Bayesian} for more details.

\section{Discussion}
\label{sec:discussion}

\cref{alg:safe_exploration} allows for the safe exploration of the ROA without leaving the ROA for a general, uncertain, nonlinear system. While the discretization in~\cref{thm:discretization} may be conservative and computing the predictions of the GP model of~$\dot{V}$ at the discretization points is computationally intensive, this needs to be compared to the generality of the statements we have made. In particular, the only assumption we have made about the unknown dynamics,~$\gctr(x)$ in~\cref{eq:dynamics_ode}, is~\cref{as:g_bounded_RKHS}, which is very general.

In the analysis, we have assumed that a method exists, which drives the system to any state~$x_\ti$ selected by~\cref{eq:acquisition_function}, without leaving the ROA. In practice, this must be further restricted by reachability properties, see~\cite{Akametalu2014Reachability}. These can be included as an additional constraint; that is, if~$\mathcal{R}$ is the safely reachable set, we select~$x_\ti$ from~$\mathcal{S}_\ti \cap \mathcal{R}$ in~\cref{eq:acquisition_function}.

\cref{alg:safe_exploration} can be made more data-efficient by only evaluating states close to the boundary of the level set, without loosing guarantees. See~\cite{Sui2015Safe} and~\cite{Berkenkamp2016Bayesian} for details.

A discrete-time variant of~\cref{alg:safe_exploration} that determines safety with~${V(f_{\pi,d}(x) + g_{\pi,d}(x)) - V(x) < 0}$ only requires measurements of the states,~$x$, rather than derivatives,~$\dot{x}$, in~\cref{as:measurements_of_deriv}. While mapping a GP model of~$\gctr(x)$ through a nonlinear Lyapunov function renders this method not analytically tractable, methods based on approximate uncertainty propagation may work well in practice.

%% file: sections/experiments.tex

\section{Experiments}

In this section. we evaluate~\cref{alg:safe_exploration} in a simulated experiment. We only provide a high-level overview of the experiment. For details refer to the documentation and source code at~\url{http://github.com/befelix/lyapunov-learning}.

We consider an inverted pendulum with angle~$\theta$, mass~${m=\unit[0.15]{kg}}$, length~${l=\unit[0.5]{m}}$, and friction coefficient~${\mu=\unit[0.05]{Nms/rad}}$. The dynamics are given by
\begin{equation}
\ddot{\theta}(t) = \frac{mgl \sin{\theta(t)} -\mu \dot{\theta}(t) + u(t)}{ml^2},
\label{eq:pendulum_ode}
\end{equation}
where~$u(t)$ is the torque applied to the pendulum. The torque is limited so that the real system cannot recover from states with~${|\theta| > \unit[30]{deg}}$. The state is~${x=(\theta, \dot{\theta})}$. We assume that we only know a linear approximation of the dynamics in~\cref{eq:pendulum_ode} for the upright equilibrium point, where additionally friction is neglected and the mass is~\unit[0.05]{kg} lighter. We use a Linear Quadratic Regulator based on this linear approximation in order to control the origin,~${x = 0}$, and use the corresponding quadratic Lyapunov function to determine the ROA of~\cref{eq:pendulum_ode}.

We model the error in~\cref{eq:pendulum_ode} as a GP, see~\cref{sec:gaussian_process}. In particular, we use the product of a linear and a Mat{\'e}rn kernel for the GP model, which encodes nonlinear functions that are two-times differentiable and take values that are bounded by a linear function from above and below with high probability. Details about kernel choice for dynamic systems can be found in~\cite{Berkenkamp2016Kernel} and a one-dimensional sample function of the kernel used here is in the upper plot of~\cref{fig:1d_example_0}. In practice, the more assumptions are made about the model error, the faster the learning algorithm will converge~\cite{Berkenkamp2016Kernel}.

In the analysis, we assumed~$\gctr$ to have bounded RKHS norm in~\cref{as:g_bounded_RKHS}. Here, we model~$\gctr$ as a sample function of the GP and set~${\beta_\ti^{1/2}=2}$ for all iterations.
We use~${\tau=0.002}$ and high-probability, local Lipschitz constants encoded by the kernel with~\cref{thm:v_dot_lipschitz}.
The initial safe set is~${ \mathcal{S}_0 = \{ (\theta, \dot{\theta}) \in \mathbb{R}^2 \,|\, |\theta| \leq 5\deg, \, |\dot{\theta}| \leq 10\deg / \mathrm{s} \} }$.

The results can be seen in~\cref{fig:inverted_pendulum}. While the prior model with the wrong mass and friction parameters estimates a safe set that is too large (includes unsafe states), \cref{alg:safe_exploration} provides a conservative estimate. As we gain more data about the dynamics and if we discretize with smaller values of~$\tau$, the estimate improves and, in the limit, converges to the true level set. Overall, \cref{alg:safe_exploration} provides a powerful tool to learn the ROA of general nonlinear systems from experiments, without leaving the safe region encoded by the Lyapunov function.

\begin{figure}
\vspace{1.5mm}
\includegraphics[scale=1]{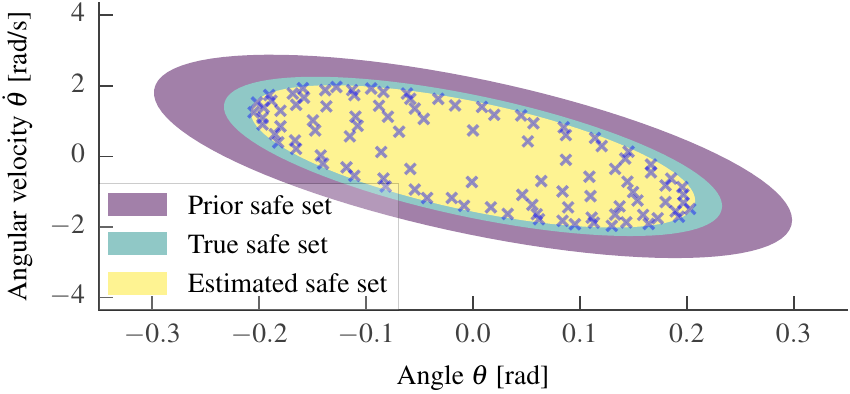}
\caption{Prior, true, and estimated ROA level sets after 100~data points (blues crosses). The prior model estimates a safe region (purple) that is larger than the true ROA (green), thus it includes unsafe states. In contrast, \cref{alg:safe_exploration} provides a conservative estimate (yellow), since it considers states with~${\dot{V}(x) \leq -L\tau}$, rather than~${\dot{V}(x) \leq 0}$. The level set could be increased by discretizing the space with a smaller value~of~$\tau$.}
\label{fig:inverted_pendulum}
\end{figure}

%% file: sections/conclusion.tex

\section{Conclusion}

We introduced a learning algorithm that, based on an initial, approximate model and a corresponding Lyapunov function, is able to learn about the real ROA through experiments on the real system, without leaving the true ROA with high probability. While some of the assumptions in~\cref{sec:prob_def} may be restrictive for practical application, the results in this paper should be understood as a theoretical foundation for learning algorithms that learn without risking instability.

%% file: sections/appendix.tex

\section*{Appendix}

\subsection{Full Exploration Proof}

We consider a variant of~\cref{alg:safe_exploration} that determines the safe set,~$\mathcal{S}_\ti$, starting from~$\mathcal{S}_0$ in~\cref{thm:deterministic_c0_stable}, using the Lipschitz constant directly. This allows us to prove exploration guarantees.

First, let us consider the best ROA approximation that we could hope to achieve. We have a probabilistic model of~$\gctr(x)$, which means we cannot expect to explore the full ROA. Instead, we consider learning the true ROA up to some error,~$\epsilon$. An algorithm with knowledge up to this error, could use the Lipschitz constant of~$\dot{V}$ from~\cref{thm:v_dot_lipschitz} in order to learn which discrete states close to~$\mathcal{S}_0$ fulfill the requirement of~$\dot{V}(x) < -L \tau$ from~\cref{thm:stability_requirement_in_X}. We define a general operator on a set that determines these states,
\begin{align}
R_{\epsilon, \dot{V}}(\mathcal{S}) = \mathcal{S} \cup \{ x \in \mathcal{X}_\tau \,|\, &\exists x' \in R_l(\mathcal{S}) \colon  \label{eq:baseline_safe_operator}\\
 &\dot{V}(x) + \epsilon + L \|x - x'\|  < -L \tau\}, \notag
\end{align}
where~$R_l(\mathcal{S})$ is an operator that selects the maximum level set of the Lyapunov function contained in~$\mathcal{S}$,
\begin{equation}
R_l(\mathcal{S}) = \mathcal{V}\left(
\argmax_{c>0}\, c, \quad \textnormal{subject to } \mathcal{V}(c) \cap \mathcal{X}_\tau \subseteq \mathcal{S}
\right).
\label{eq:level_set_operator}
\end{equation}
Thus,~\cref{eq:baseline_safe_operator} adds a state~$x \in \mathcal{X}_\tau$ to $\mathcal{S}$ if there exists a state~$x'$ in the current level set estimate of the ROA that can be used in order to determine that~${\dot{V}(x) < -L \tau}$ via the Lipschitz constant. If we apply this operator iteratively, in the limit we reach the maximum ROA that can be determined using knowledge up to~$\epsilon$. That is, with~$R_{\epsilon, \dot{V}}^i(\mathcal{S}) = R_{\epsilon, \dot{V}}(R_{\epsilon, \dot{V}}^{i-1}(\mathcal{S}))$ and $R_{\epsilon, \dot{V}}^1(\mathcal{S}) = R_{\epsilon, \dot{V}}(\mathcal{S})$, the maximum safe set that could be determined is
\begin{equation}
R_l(\overline{R}_{\epsilon, \dot{V}}(\mathcal{S}_0)), \textnormal{~~where~~} \overline{R}_{\epsilon, \dot{V}}(\mathcal{S}) = \lim_{i \to \infty} R^i_{\epsilon, \dot{V}}(\mathcal{S}).
\end{equation}

With the baseline defined, in the following we provide an algorithm that achieves this baseline using the results from~\cref{sec:lyapunov}. Instead of defining the ROA directly in terms of the GP confidence interval of~$\dot{V}$,~${\mathcal{Q}_{\ti, \dot{V}}(x) = [l_\ti(x),\, u_\ti(x)]}$ as in~\cref{thm:roa_in_S}, we consider the intersection of these intervals, ${\mathcal{C}_{\ti}(x) = \mathcal{Q}_{\ti,\dot{V}}(x) \cap \mathcal{C}_{\ti-1}(x) }$. We initialize these intervals so that points in~$\mathcal{S}_0$ are safe; that is,~${\mathcal{C}_0(x) = [-\infty,\, 0)}$ if~${x \in \mathcal{S}_0}$ and~${[-\infty,\, \infty]}$ otherwise. As a consequence of~\cref{thm:confidence_interval}, we have that~${\dot{V}(x) \in \mathcal{C}_\ti(x)}$ for all~${\ti \geq 1}$ with probability at least~${(1 - \delta)}$. Besides ensuring that points in~$\mathcal{S}_0$ are always considered safe, this definition guarantees that the estimated safe set will not decrease over time. We define~${l_{\ti, c}(x) = \min_{x \in \mathcal{X}} \mathcal{C}_\ti(x) }$ and ${u_{\ti, c}(x) = \max_{x \in \mathcal{X}} \mathcal{C}_\ti(x) }$.

Based on these confidence intervals, after~${(\ti-1)}$ measurements we can quantify the following data points as having a value~$\dot{V}(x)$ that is sufficiently negative,
\begin{align}
\mathcal{S}_{\ti, \dot{V}} = \mathcal{S}_0 \cup \{ x \in \mathcal{X} \,|\, &\exists x' \in R_l(S_{\ti - 1, \dot{V}}) \colon
\label{eq:gp_safe_set_definition}\\
&u_{\ti, c}(x') + L \|x - x' \| < -L\tau \}, \notag
\end{align}
and estimate the ROA based on this set as
\begin{equation}
\mathcal{S}_\ti = R_l(\mathcal{S}_{\ti, \dot{V}}).
\label{eq:lipschitz_safe_set}
\end{equation}
Intuitively this selection criterion is similar to the one of the baseline in~\cref{eq:baseline_safe_operator}, except that~\cref{eq:gp_safe_set_definition} uses the GP confidence intervals instead of perfect model knowledge. This allows us to prove that we explore the maximum level set up to~$\epsilon$ accuracy:
\begin{theorem}
\label{thm:full_exploration}
Under the assumptions of~\cref{thm:roa_in_S}, choose~$\beta_\ti$ as in~\cref{thm:confidence_interval} and let~$\ti^*$ be the smallest positive integer so that
\begin{equation}
\frac{\ti^*}{\beta_{\ti^*} \gamma_{\ti^*}} \geq \frac{C L_{\partial V}^2 (| \overline{R}_{0, \dot{V}}(\mathcal{S}_0) | + 1)}{\epsilon^2},
\end{equation}
where~${C = 8 / \log(1 + \sigma^{-2})}$. For any~${\epsilon > 0}$, and~${\delta \in (0,\,1)}$, under the selection criterion~\cref{eq:acquisition_function} within~$\mathcal{S}_\ti$ from~\cref{eq:lipschitz_safe_set}, the following holds jointly with probability at least~${(1-\delta)}$:
\begin{enumerate}[(i)]
\item $\mathcal{S}_\ti$ is contained in the ROA of~\cref{eq:dynamics_ode} for all~${\ti \geq 1}$,
\item $R_l(\overline{R}_{\epsilon, \dot{V}}(\mathcal{S}_0)) \subseteq \mathcal{S}_{\ti^*} \subseteq R_l(\overline{R}_{0, \dot{V}}(\mathcal{S}_0))$.
\end{enumerate}
\end{theorem}
\begin{proof}
Statement~\textit{(i)} follows from~\cref{thm:roa_in_S}. For \textit{(ii)}, we have from \cite[Lemma 5]{Sui2015Safe} that if the safe set does not increase, then the maximum uncertainty of~$\gctr$ and~$\dot{V}$ within the safe set is bounded by~${ \epsilon / \| \partial V / \partial x \|_\infty = \epsilon / L_{\partial V} }$ and~$\epsilon$, respectively, after a finite number of iterations. At that point, either the safe set increases similarly to the baseline, or we have explored the full safe set~\cite[Lemma 7]{Sui2015Safe}. Applying this ${\left(|\overline{R}_{0, \dot{V}}(\mathcal{S}_0) | + 1\right)}$ times provides the result~\cite[Cor. 4]{Sui2015Safe}.
\end{proof}

That is, after a finite number of evaluations,~$\ti^*$, we explore at least as much as the baseline up to accuracy~$\epsilon$, but not more than we could determine as safe if we knew the function perfectly; that is, the baseline with~$\epsilon=0$.